\documentclass[11pt]{article}
\usepackage{latexsym}
\usepackage{amsfonts}
\usepackage{mathrsfs,amsthm,amsmath}
\usepackage{color}
\usepackage[round]{natbib}
\newtheorem{theorem}{Theorem}[section]
\newtheorem{proposition}[theorem]{Proposition}

\newtheorem{remark}{Remark}[section]
\newtheorem{example}[theorem]{Example}
\newtheorem{condition}{Condition}[section]
\begin{document}
\title{Large deviations for risk measures in finite mixture
models\thanks{The support of Gruppo Nazionale per l'Analisi
Matematica, la Probabilit\`{a} e le loro Applicazioni (GNAMPA) of
the Istituto Nazionale di Alta Matematica (INdAM) is
acknowledged.}}
\author{Valeria Bignozzi\thanks{Dipartimento di Statistica e Metodi
Quantitativi, Universit\`{a} di Milano Bicocca, Via Bicocca degli
Arcimboldi 8, I-20126 Milano, Italia. e-mail:
\texttt{valeria.bignozzi@unimib.it}}\and Claudio
Macci\thanks{Dipartimento di Matematica, Universit\`a di Roma Tor
Vergata, Via della Ricerca Scientifica, I-00133 Roma, Italia.
e-mail: \texttt{macci@mat.uniroma2.it}}\and Lea
Petrella\thanks{Dipartimento di Metodi e Modelli per l'Economia,
il Territorio e la Finanza, Sapienza Universit\`{a} di Roma, Via
del Castro Laurenziano 9, I-00161 Roma, Italia. e-mail:
\texttt{lea.petrella@uniroma1.it}}}
\date{}
\maketitle
\begin{abstract}
\noindent Due to their heterogeneity, insurance risks can be
properly described as a mixture of different fixed models, where
the weights assigned to each model may be estimated empirically
from a sample of available data. If a risk measure is evaluated on
the estimated mixture instead of the (unknown) true one, then it
is important to investigate the committed error. In this paper we
study the asymptotic behaviour of estimated risk measures, as the
data sample size tends to infinity, in the fashion of large
deviations. We obtain large deviation results by applying the
contraction principle, and the rate functions are given by a
suitable variational formula; explicit expressions are available
for mixtures of two models. Finally, our results are applied to
the most common risk measures, namely the quantiles, the Expected
Shortfall and the shortfall risk measure.
\end{abstract}

\section{Introduction}
Quantitative risk management for financial and insurance companies
requires the modelling of financial positions in terms of random
variables on a suitable probability space; in mathematical terms,
this corresponds to identifying a probability law (model) $\mu$ on
the real line that describes as accurately as possible, the random
behaviour of the position. Model risk, that arises from the
uncertainty about the model to adopt, has been largely discussed
in various area of the literature, because it may impact
substantially companies decision making and performance. We can
distinguish three main approaches to deal with model uncertainty:
$1)$ the model is not specified but directly extrapolated from
data via the empirical distribution; $2)$ a model is selected and
its parameters are estimated from data (e.g. using Maximum
Likelihood Estimation); $3)$ a class of candidate models is
considered (for instance models suggested by expert opinion) and
then one or an average of them is applied. The latter approach is
probably the most common one and includes for instance: the
worst-case approach proposed by \cite{GS89} in the theory of
utility maximization, where the chosen model is the one providing
the most adverse outcome; the Bayesian model averaging approach,
developed by \cite{RMH97} where (posterior) weights are calculated
for each model considering both information arising from data and
prior beliefs; the highest posterior approach, where the selected
model is the one most favourable according to the posterior
weights. \cite{C00} provided a general framework for dealing with
model and parameter uncertainty in an insurance framework, while
\cite{PSZ09} considered a model averaging approach in a risk
management context. As we will see later, in this contribution we
consider an average of fixed models where the weights are
estimated empirically.

A second fundamental step for internal and external risk
management purposes is to quantify the riskiness of the company
positions. Once the model has been chosen, this essentially
corresponds to applying a suitable risk measure $\rho$ to the
financial position. The impact of model uncertainty on risk
measurement was discussed among others by \cite{BS15} and
\cite{BT16} where different measures of model risk are considered.

Most of the risk measures generally considered, from both
academics and practitioners, are \textit{law-invariant} that is,
univocally determined by the probability law of the random
variable. These risk measures can then be treated as statistical
functionals.

While the mathematical theory of risk measures is by now well
developed, we refer for instance to \cite{FS16} for an extensive
treatment of coherent and convex risk measures, research on the
statistical properties of risk measures is fairly recent. The
seminal paper by \cite{CDS10} started a new strand in the
literature that investigates the statistical properties of risk
measures in terms of robustness with respect to available data and
to different model estimation procedures. The main difference
between the mathematical and the statistical approaches is that,
in the first case risk measures are defined on a space of random
variables, while in the second one, on a space of probability
measures. Although, under weak technical assumptions, for a random
variable $X$ with probability law $\mu$, we can identify $\rho(X)$
and $\rho(\mu)$, it is important to emphasise that properties of
risk measures on random variables and on distributions are
different. In particular, given two random variables $X$, $Y$ with
distributions $\mu$, $\nu$ the convex combination
$$p X+(1-p)Y,$$
for $p\in(0,1)$, corresponds to a diversification of the
portfolio, therefore it is reasonable to require a risk measure to
be convex with respect to random variables. On the other side, the
mixture distribution
$$p\mu+(1-p)\nu$$
represents a higher risk profile and thus a risk measure should
not be convex with respect to mixtures of distributions.
Properties of the risk measures with respect to mixture
distributions, have been investigated by \cite{AS13}. \cite{W06}
used such properties to characterise dynamic risk measures, while
\cite{Z16},~\cite{BB15} and \cite{DBBZ16} used them to study
elicitable functionals. \cite{BMP17} presented some results on
risk measures evaluated on mixtures of Gaussian and Student $t$
distributions.

In this contribution we consider risk measures applied to the
mixture distribution
$$\pi_1\mu_1+\ldots+\pi_s\mu_s,$$
where $\{\mu_1,\ldots,\mu_s\}$ is a set of $s$ available models,
and $\pi_1,\ldots,\pi_s\geq 0$ (with $\sum_{j=1}^s\pi_j=1$) are
the weights assigned to each model. Mixture models are
particularly relevant when a single model is not sufficient to
fully describe the data. They represent a flexible approach for
modelling heterogeneous data and to carry out cluster analysis.
Further, mixture models represent a ductile way to model unknown
distributional shapes. Such situations are quite common in
insurance where often a mix of small, medium and large size claims
occurs; we refer the interested reader to \cite{KPW12} for a full
treatment of loss modelling in actuarial science. \cite{BMP12}
proposed finite mixtures of Skew Normal distributions to properly
characterise insurance data, while \cite{LL10} suggested a mixture
of Erlang distributions. In a statistical framework mixture models
have a variety of applications; we refer to \cite{MP04} for an
extensive treatment of the topic.

Throughout this paper the models $\mu_1,\ldots,\mu_s$ are assumed
to be fixed, and the weights $\pi_1,\ldots,\pi_s$ are estimated
empirically from independent samples. In an insurance framework,
we can assume that each model represents the loss profile of a
customer (or a class of customers) and the weights are estimated
registering the relative frequency of claims occurring for each
model. Then we consider the sequence of \emph{empirical} \emph{risk}
\emph{measures}  $\left\{\rho(\sum_{j=1}^s\hat{\pi}_n(j)\mu_j):n\geq
1\right\}$, where the weight estimators
$\hat{\pi}_n(1),\ldots,\hat{\pi}_n(s)$ concern the empirical law
of i.i.d. random variables $\{X_1,\ldots,X_n\}$ with distribution
$\pi=(\pi_1,\ldots,\pi_s)$ (see \eqref{eq:def-weight-estimators}
below).

In this paper we prove large deviation results for the empirical
risk measures. The theory of large deviations gives an asymptotic
computation of small probabilities on an exponential scale (see
e.g. \citealp{DemboZeitouni} as a reference on this topic). The
large deviation principles are obtained by applying the
contraction principle; so the rate functions are given by a
suitable variational formula. We use the method of the Lagrange
multipliers, and explicit expressions are available for $s=2$. We
then apply our results to the most common risk measures, namely
the quantiles (also known as Value-at-Risk in the risk management
literature), the Expected Shortfall (ES) and the shortfall risk
measure.

A different approach for large deviation analysis may be the use
of precise large deviation techniques which are beyond the purpose
of the paper; among others, a possible reference for the
interested reader is \cite{FMN}.

Our work was inspired by \cite{W07}, where the author considered
the empirical risk measures $\left\{\rho(\hat{\mu}_n):n\geq
1\right\}$, and
$$\hat{\mu}_n:=\frac{1}{n}\sum_{i=1}^n\delta_{Y_i}$$
is the empirical law of i.i.d.  random variables
$\{Y_1,\ldots,Y_n\}$ having (unknown) distribution $\mu$ with
bounded support. The main goal of that paper is to investigate
coherent and convex risk measures that are continuous on compacts.
This condition  yields the large deviation principle of
$\left\{\rho(\hat{\mu}_n):n\geq 1\right\}$ by applying the
contraction principle (see Proposition 2.1 and Corollary 2.1 in
\citealp{W07}). 

The paper is organised as follows. Section \ref{sec:preliminaries}
gathers some preliminaries on large deviations and their
applications to our framework with finite mixtures. In Section
\ref{sec:results} we present the main results of the paper, while
Section \ref{sec:examples} presents some examples for the most
common risk measures used in practice and in the literature.

\section{Preliminaries}\label{sec:preliminaries}
In this section we recall some preliminaries on large deviations
and a large deviation principle for a sequence of estimators (see
Proposition \ref{prop:analogue-of-Prop.2.1-Weber}).

\subsection{Preliminaries on large deviations}
A sequence of random variables $\{W_n:n\geq 1\}$ taking values on
a topological space $\mathcal{W}$ satisfies the large deviation
principle (LDP for short) with rate function
$I:\mathcal{W}\to[0,\infty]$ if $I$ is a lower semi-continuous
function,
$$\liminf_{n\to\infty}\frac{1}{n}\log P(W_n\in O)\geq-\inf_{w\in O}I(w)\ \mbox{for all open sets}\ O$$
and
$$\limsup_{n\to\infty}\frac{1}{n}\log P(W_n\in C)\leq-\inf_{w\in C}I(w)\ \mbox{for all closed sets}\ C.$$
A rate function $I$ is said to be good if all its level sets
\mbox{$\{\{w\in\mathcal{W}:I(w)\leq\eta\}:\eta\geq 0\}$} are
compact. Finally we also recall the contraction principle (see
e.g.~Theorem 4.2.1 in \citealp{DemboZeitouni}): let $\mathcal{Y}$ be
a topological space, and let $f:\mathcal{W}\to\mathcal{Y}$ be a
continuous function; then, if $\{W_n:n\geq 1\}$ satisfies the LDP
with good rate function $I$, and $Y_n:=f(W_n)$ (for all $n\geq
1$), $\{Y_n:n\geq 1\}$ satisfies the LDP with good rate function
$J$ defined by
$$J(y):=\inf\{I(w):w\in\mathcal{W},f(w)=y\}.$$

The LDP for real valued random variables is used to obtain
asymptotic evaluations for the logarithm of tail probabilities;
indeed, for a wide class of cases, we have
$$\log P(W_n>x)\sim {-nI(x)}$$
at least for $x$ large enough to have $I(x)=\inf_{w>x}I(w)$ (we
use the symbol $\sim$ to mean that the ratio tends to 1 as
$n\to\infty$).

\subsection{LDP for estimators of $\rho(\mu)$ when $\mu$ is a mixture}
We define a law-invariant risk measure as a map
$$\rho:\mathcal{P}(\mathbb{R})\to \mathbb{R},$$
that assigns to every probability measure
$\mu\in\mathcal{P}(\mathbb{R})$ on the real line a real number
$\rho(\mu)$. Such a value, is generally used to summarise the
riskiness of the model $\mu$ and can be adopted to calculate
solvency capital requirements. In the present contribution, we
focus on probability distributions that arise as mixtures
$\pi_1\mu_1+\ldots+\pi_s\mu_s$ of some fixed models
$\mu_1,\ldots,\mu_s$  with weights
$\pi=(\pi_1,\ldots,\pi_s)\in\Sigma_s$ where
$$\Sigma_s:=\{(p_1,\ldots,p_s):p_1,\ldots,p_s\geq 0,\ p_1+\cdots+p_s=1\}$$
is the simplex; we are then interested in  computing
$\rho\left(\sum_{j=1}^s\pi_j\mu_j\right)$. In mixture models used
for modeling insurance data, it is often the case that the weights
$\pi_1,\ldots,\pi_s$ are unknown   and estimated from a set of $n$
available data by
$\hat{\pi}_n=(\hat{\pi}_n(1),\ldots,\hat{\pi}_n(s))$, see for
instance \cite{LLW12}.  In order to estimate the error committed
in computing the estimated risk measure
$\rho\left(\sum_{j=1}^s\hat{\pi}_n(j)\mu_j\right)$ instead of the
correct one $\rho\left(\sum_{j=1}^s{\pi}(j)\mu_j\right)$, we
employ the theory of large deviations. In particular, we consider
the case where  the weights are estimated empirically as
\begin{equation}\label{eq:def-weight-estimators}
\hat{\pi}_n(j):=\frac{1}{n}\sum_{i=1}^n\delta_{X_i=j}\ (\mbox{for
all}\ j\in\{1,\ldots,s\}),
\end{equation}
where $\{X_1,\ldots,X_n\}$ are i.i.d. random variables with
distribution $\pi=(\pi_1,\ldots,\pi_s)$. It is well known that
the  sequence of empirical measures $(\hat{\pi}_n)_n$  converges
$P$-a.s. to $\pi$, and that it satisfies the LDP (see e.g. Theorem
2.1.10 in \citealp{DemboZeitouni}). Therefore, by applying the
contraction principle (see Theorem 4.2.1 in
\citealp{DemboZeitouni}), we obtain the LDP stated in the
following proposition.

\begin{proposition}\label{prop:analogue-of-Prop.2.1-Weber}
Let $\pi=(\pi_1,\ldots,\pi_s)\in\Sigma_s$. Moreover assume
that the function
\begin{equation}\label{eq:function-of-weights}
\Sigma_s\ni(p_1,\ldots,p_s)\mapsto\rho\left(\sum_{j=1}^sp_j\mu_j\right)
\end{equation}
is continuous. Then
$\left\{\rho\left(\sum_{j=1}^s\hat{\pi}_n(j)\mu_j\right):n\geq
1\right\}$ satisfies the LDP (as $n\to\infty$) with good rate
function $H_{\rho,\langle\pi,\mu\rangle}$ defined by
\begin{equation}\label{eq:rf-variational-formula}
H_{\rho,\langle\pi,\mu\rangle}(r):=\inf\left\{\sum_{j=1}^sp_j\log\frac{p_j}{\pi_j}:(p_1,\ldots,p_s)\in\Sigma_s,\rho\left(\sum_{j=1}^sp_j\mu_j\right)=r\right\}.
\end{equation}
\end{proposition}

\begin{remark}[Relative entropy and Sanov's Theorem]\label{rem:relative-entropy}
The quantity
$\mathbb{E}_{\pi}[\frac{dp}{d\pi}\log\frac{dp}{d\pi}]=\sum_{j=1}^sp_j\log\frac{p_j}{\pi_j}$
in \eqref{eq:rf-variational-formula} is the relative entropy of a
general probability measure $p=(p_1,\ldots,p_s)$ on the state
space $\{1,\ldots,s\}$, with respect to the probability measure
$\pi=(\pi_1,\ldots,\pi_s)$ that gives the actual (but unknown)
weights of the mixture model. Large deviation rate functions are
indeed often expressed in terms of relative entropy; see e.g. the
discussion in \cite{Varadhan}. The rate function is thus obtained
minimising the relative entropy under a constraint on the risk
measure.
\end{remark}

\begin{remark}[The set $S_\pi$ and the value $r_0$]\label{rem:initial}
If $\pi_j=0$ for some $j\in\{1,\ldots,s\}$, then we have
$$p_j\log\frac{p_j}{\pi_j}=\left\{\begin{array}{ll}
0&\ \mbox{if}\ p_j=0\\
\infty&\ \mbox{if}\ p_j\in(0,1];
\end{array}\right.$$
so, in some sense, the index $j$ is negligible. Then we should
consider the set $S_\pi:=\{i\in\{1,\ldots,s\}:\pi_i>0\}$ instead
of $\{1,\ldots,s\}$; however, with a slight abuse of notation,
throughout the paper we always refer to $\{1,\ldots,s\}$ (and its
cardinality $s$) because we can always rearrange the notation in
order to have $S_\pi=\{1,\ldots,s\}$. We also remark that
$H_{\rho,\langle\pi,\mu\rangle}(r)$ uniquely vanishes at $r=r_0$,
where
$$r_0:=\rho\left(\sum_{j=1}^s\pi_j\mu_j\right).$$
Thus we can say that, for every $\delta>0$, under the hypotheses
of Proposition \ref{prop:analogue-of-Prop.2.1-Weber}, the
probability
$$P\left(\left|\rho\left(\sum_{j=1}^s\hat{\pi}_n(j)\mu_j\right)-r_0\right|\geq\delta\right)$$
decays as $e^{-nh_\delta}$, where
$h_\delta:=\inf\left\{H_{\rho,\langle\pi,\mu\rangle}(r):
|r-r_0|\geq\delta\right\}>0$, as $n\to\infty$.
\end{remark}

\section{Results}\label{sec:results}
In this section  we  provide, when possible,
 an explicit expression of the variational formula in
\eqref{eq:rf-variational-formula}. Note  that in general the constraint $\rho\left(\sum_{j=1}^sp_j\mu_j\right)=r$ cannot be written explicitly in terms of the $p_j$'s; for this reason
we introduce the next Condition
\ref{cond:equivalent-formulation-on-constraint} that requires a sort of  linear dependence of the risk measures with respect to  the mixture weights. As we shall
see in Theorem \ref{prop:main}, this allows us to handle the variational formula in
\eqref{eq:rf-variational-formula}   with
 the method of Lagrange multipliers. Condition
\ref{cond:equivalent-formulation-on-constraint} does not seem to
be restrictive, it is indeed satisfied by many of the risk
measures used by academics and  practitioners.

\begin{condition}\label{cond:equivalent-formulation-on-constraint}
The function in \eqref{eq:rf-variational-formula} can be written
as
\begin{equation}\label{eq:variational-formula-special-constraint}
H_{\rho,\langle\pi,\mu\rangle}(r):=\inf\left\{\sum_{j=1}^sp_j\log\frac{p_j}{\pi_j}:(p_1,\ldots,p_s)\in\Sigma_s,\sum_{j=1}^sp_j\Psi_\rho(\mu_j,r)=0\right\},
\end{equation}
for some (strictly) decreasing functions
$\Psi_\rho(\mu_1,\cdot),\ldots,\Psi_\rho(\mu_s,\cdot)$. Moreover,
for all $i\in\{1,\ldots,s\}$, there exists a unique $r_i^{(0)}$
such that $\Psi_\rho\left(\mu_i,r_i^{(0)}\right)=0$.
\end{condition}

Obviously we could require that
$\Psi_\rho(\mu_1,\cdot),\ldots,\Psi_\rho(\mu_s,\cdot)$ are
increasing functions (instead of decreasing); in such a case we
can reduce to Condition
\ref{cond:equivalent-formulation-on-constraint} (namely the
functions $\Psi_\rho(\mu_1,\cdot),\ldots,\Psi_\rho(\mu_s,\cdot)$
are decreasing) by a change of sign.

We will see in Section \ref{sec:examples} that Condition
\ref{cond:equivalent-formulation-on-constraint} is fulfilled by
some popular risk measures, such as the quantiles, the mean and
the class of convex shortfall risk measures introduced by
\cite{FS02}. The Expected Shortfall satisfies this condition only
under some extra requirements. A similar condition recently
appeared in the literature about elicitable risk measures under
the name of Convex Level Sets (CxLS). A risk measure has CxLS if,
given $\rho(\mu_1)=\cdots=\rho(\mu_s)=r$, then
$\rho(\sum_{j=1}^sp_j\mu_j)=r$. Clearly, if
$\rho(\mu_1)=\cdots=\rho(\mu_s)=r$, the CxLS property implies our
Condition \ref{cond:equivalent-formulation-on-constraint} with
$$\Psi_\rho(\mu_j,r):=\rho(\mu_j)-r.$$
A full characterization of convex risk measures satisfying the
CxLS property is provided in \cite{DBBZ16}.

\begin{remark}[Consequences of Condition \ref{cond:equivalent-formulation-on-constraint} for $r_0$]\label{rem:consequence-condition}
If Condition \ref{cond:equivalent-formulation-on-constraint}
holds, then we have
\begin{equation}\label{eq:cond-on-pi-and-r0}
\sum_{j=1}^s\pi_j\Psi_\rho(\mu_j,r_0)=0.
\end{equation}
Moreover, if we set
$\underline{r}_{\rho}:=\min\left\{r_i^{(0)}:i\in\{1,\ldots,s\}\right\}$
and
$\overline{r}_{\rho}:=\max\left\{r_i^{(0)}:i\in\{1,\ldots,s\}\right\}$,
we have $\underline{r}_{\rho}\leq\overline{r}_{\rho}$. Then we
can distinguish two cases (see parts (i) and (ii) in the next Theorem
\ref{prop:main}):
\begin{itemize}
\item $\underline{r}_{\rho}=\overline{r}_{\rho}=:\hat{r}_\rho$,
which occurs if and only if
$r_1^{(0)}=\cdots=r_s^{(0)}=\hat{r}_\rho$; in this case we have
$r_0=\hat{r}_\rho$, and the estimators
$\left\{\rho\left(\sum_{j=1}^s\hat{\pi}_n(j)\mu_j\right);n\geq
1\right\}$ are constantly equal to $r_0$;
\item $\underline{r}_{\rho}<\overline{r}_{\rho}$; in this case
we have $r_0\in(\underline{r}_{\rho},\overline{r}_{\rho})$, and
the estimators
$\left\{\rho\left(\sum_{j=1}^s\hat{\pi}_n(j)\mu_j\right);n\geq
1\right\}$ take values in
$[\underline{r}_{\rho},\overline{r}_{\rho}]$.
\end{itemize}
The first case always occurs if $s=1$.
\end{remark}

Now we are ready to present Theorem \ref{prop:main}. In general we
only have an explicit expression of
$H_{\rho,\langle\pi,\mu\rangle}$ for the case $s=2$; see Remark
\ref{rem:not-explicit} and Remark \ref{rem:case-s=2}. The case
with $s=\infty$ will be discussed in Remark
\ref{rem:case-s=infty}.

\begin{theorem}\label{prop:main}
Consider the same hypotheses of Proposition
\ref{prop:analogue-of-Prop.2.1-Weber}. Assume that Condition
\ref{cond:equivalent-formulation-on-constraint} holds, and let
$\underline{r}_{\rho}$ and $\overline{r}_{\rho}$ be as in Remark
\ref{rem:consequence-condition}.\\
(i) If $\underline{r}_{\rho}=\overline{r}_{\rho}$, then
$$H_{\rho,\langle\pi,\mu\rangle}(r)=\left\{\begin{array}{ll}
0&\ \mbox{if}\ r=r_1^{(0)}=\cdots=r_s^{(0)}\\
\infty&\ \mbox{otherwise}.
\end{array}\right.$$
(ii) If $\underline{r}_{\rho}<\overline{r}_{\rho}$, then
$$H_{\rho,\langle\pi,\mu\rangle}(r)=\left\{\begin{array}{ll}
-\log\left(\sum_{j=1}^s\pi_je^{-\lambda_*(r)\Psi_\rho(\mu_j,r)}\right)&\ \mbox{if}\ r\in(\underline{r}_{\rho},\overline{r}_{\rho})\\
-\log\sum_{j:r_j^{(0)}=\overline{r}_{\rho}}\pi_j&\ \mbox{if}\ r=\overline{r}_{\rho}\\
-\log\sum_{j:r_j^{(0)}=\underline{r}_{\rho}}\pi_j&\ \mbox{if}\ r=\underline{r}_{\rho}\\
\infty&\ \mbox{if}\
r\notin[\underline{r}_{\rho},\overline{r}_{\rho}],
\end{array}\right.$$
where $\lambda_*(r)$ is such that
\begin{equation}\label{eq:constraint-lambda}
\frac{\sum_{j=1}^s\pi_j\Psi_\rho(\mu_j,r)e^{-\lambda_*(r)\Psi_\rho(\mu_j,r)}}{\sum_{j=1}^s\pi_je^{-\lambda_*(r)\Psi_\rho(\mu_j,r)}}=0.
\end{equation}
\end{theorem}
\begin{proof}
We start with the proof of the statement (i). For
$r=r_1^{(0)}=\cdots=r_s^{(0)}$ we have
$$H_{\rho,\langle\pi,\mu\rangle}(r)=\inf\left\{\sum_{j=1}^sp_j\log\frac{p_j}{\pi_j}:(p_1,\ldots,p_s)\in\Sigma_s\right\}=0$$
(the infimum is attained by choosing
$(p_1,\ldots,p_s)=(\pi_1,\ldots,\pi_s)$); on the contrary, for
$r\neq r_1^{(0)}=\cdots=r_s^{(0)}$, we have
$H_{\rho,\langle\pi,\mu\rangle}(r)=\infty$ because the condition
$\sum_{j=1}^sp_j\Psi_\rho(\mu_j,r)=0$ fails for every choice of
$(p_1,\ldots,p_s)\in\Sigma_s$ (in fact the values
$\{\Psi_\rho(\mu_j,r):j\in\{1,\ldots,s\}\}$ are all positive if
$r<r_1^{(0)}=\cdots=r_s^{(0)}$ and are all negative if
$r>r_1^{(0)}=\cdots=r_s^{(0)}$), and therefore we have the infimum
over the empty set.

Now we concentrate the attention on the proof of the statement
(ii). For $r\notin[\underline{r}_{\rho},\overline{r}_{\rho}]$ we
have the same argument of the proof of the statement (i), for the
case $r\neq r_1^{(0)}=\cdots=r_s^{(0)}$. For
$r\in(\underline{r}_{\rho},\overline{r}_{\rho})$ we introduce the
function
$$\mathcal{L}(p_1,\ldots,p_s,\lambda)=\sum_{j=1}^sp_j\log\frac{p_j}{\pi_j}+\lambda\left(\sum_{j=1}^sp_j\Psi_\rho(\mu_j,r)\right)$$
and, by the Lagrange multipliers method, $(p_1,\ldots,p_s)$
attains the infimum in
\eqref{eq:variational-formula-special-constraint} if it is the
solution of the system
$$\left\{\begin{array}{ll}
\log\frac{p_i}{\pi_i}+1+\lambda\Psi_\rho(\mu_i,r)=0,\qquad\ \mbox{for all}\ i\in\{1,\ldots,s\}\\
\sum_{j=1}^sp_j\Psi_\rho(\mu_j,r)=0;
\end{array}\right.$$
then the minimiser $(p_1,\ldots,p_s)=(p_1(r),\ldots,p_s(r))$,
which attains the infimum in
\eqref{eq:variational-formula-special-constraint}, is defined by
\begin{equation}\label{eq:def-minimiser-special-constraint}
p_i(r):=\frac{\pi_ie^{-\lambda_*(r)\Psi_\rho(\mu_i,r)}}{\sum_{j=1}^s\pi_je^{-\lambda_*(r)\Psi_\rho(\mu_j,r)}}\
\mbox{for all}\ i\in\{1,\ldots,s\},
\end{equation}
where $\lambda_*(r)$ is such that \eqref{eq:constraint-lambda}
holds. In conclusion, for
$r\in(\underline{r}_{\rho},\overline{r}_{\rho})$, we have
\begin{multline*}
H_{\rho,\langle\pi,\mu\rangle}(r)=\sum_{j=1}^sp_j(r)\log\frac{p_j(r)}{\pi_j}
=\sum_{j=1}^s\frac{\pi_je^{-\lambda_*(r)\Psi_\rho(\mu_j,r)}}{\sum_{h=1}^s\pi_he^{-\lambda_*(r)\Psi_\rho(\mu_h,r)}}
\log\frac{e^{-\lambda_*(r)\Psi_\rho(\mu_j,r)}}{\sum_{h=1}^s\pi_he^{-\lambda_*(r)\Psi_\rho(\mu_h,r)}}\\
=-\lambda_*(r)\underbrace{\sum_{j=1}^s\frac{\pi_j\Psi_\rho(\mu_j,r)e^{-\lambda_*(r)\Psi_\rho(\mu_j,r)}}
{\sum_{h=1}^s\pi_he^{-\lambda_*(r)\Psi_\rho(\mu_h,r)}}}_{=0\
\mathrm{by}\ \eqref{eq:constraint-lambda}}
-\log\left(\sum_{h=1}^s\pi_he^{-\lambda_*(r)\Psi_\rho(\mu_h,r)}\right),
\end{multline*}
and therefore
\begin{equation}\label{eq:rf-expression-special-constraint}
H_{\rho,\langle\pi,\mu\rangle}(r)=-\log\left(\sum_{h=1}^s\pi_he^{-\lambda_*(r)\Psi_\rho(\mu_h,r)}\right).
\end{equation}
Finally the cases
$r\in\{\underline{r}_{\rho}, \overline{r}_{\rho}\}$. The minimisers
$(p_1,\ldots,p_s)=(p_1(r),\ldots,p_s(r))$, which attain the
infimum in \eqref{eq:variational-formula-special-constraint}, are
$$p_i(\overline{r}_{\rho})=\left\{\begin{array}{ll}
\pi_i/(\sum_{j:r_j^{(0)}=\overline{r}_{\rho}}\pi_j)&\ \mbox{if}\ i\in\{j:r_j^{(0)}=\overline{r}_{\rho}\}\\
0&\ \mbox{if}\ i\notin\{j:r_j^{(0)}=\overline{r}_{\rho}\}
\end{array}\right.$$
and
$$p_i(\underline{r}_{\rho})=\left\{\begin{array}{ll}
\pi_i/(\sum_{j:r_j^{(0)}=\underline{r}_{\rho}}\pi_j)&\ \mbox{if}\ i\in\{j:r_j^{(0)}=\underline{r}_{\rho}\}\\
0&\ \mbox{if}\ i\notin\{j:r_j^{(0)}=\underline{r}_{\rho}\},
\end{array}\right.$$
respectively; thus we can easily check that
$$H_{\rho,\langle\pi,\mu\rangle}(r)=\sum_{j=1}^sp_j(r)\log\frac{p_j(r)}{\pi_j}=\left\{\begin{array}{ll}
-\log\sum_{j:r_j^{(0)}=\overline{r}_{\rho}}\pi_j&\ \mbox{if}\ r=\overline{r}_{\rho}\\
-\log\sum_{j:r_j^{(0)}=\underline{r}_{\rho}}\pi_j&\ \mbox{if}\
r=\underline{r}_{\rho}.
\end{array}\right.$$
\end{proof}

\begin{remark}[The rate function $H_{\rho,\langle\pi,\mu\rangle}$ is not explicit]\label{rem:not-explicit}
Obviously Theorem \ref{prop:main} does not provide an explicit
expression of the rate function because there is not an explicit
expression of $\lambda_*(r)$. However
$H_{\rho,\langle\pi,\mu\rangle}(r_0)=0$; so that
$(p_1(r_0),\ldots,p_s(r_0))=(\pi_1,\ldots,\pi_s)$ and, by taking
into account \eqref{eq:def-minimiser-special-constraint}, we obtain
$\lambda_*(r_0)=0$. Therefore we recover the known equality
$H_{\rho,\langle\pi,\mu\rangle}(r_0)=0$ by considering
\eqref{eq:rf-expression-special-constraint} and
$\lambda_*(r_0)=0$. Finally we also remark that
\eqref{eq:constraint-lambda} and $\lambda_*(r_0)=0$ yield
\eqref{eq:cond-on-pi-and-r0}.
\end{remark}

\begin{remark}[On Theorem \ref{prop:main}(ii) with $s=2$]\label{rem:case-s=2}
We remark that, if $s=2$, then
$\underline{r}_{\rho}<\overline{r}_{\rho}$ if and only if
$r_1^{(0)}\neq r_2^{(0)}$. Here we take
$r\in(\underline{r}_{\rho},\overline{r}_{\rho})=\left(r_1^{(0)}\wedge
r_2^{(0)},r_1^{(0)}\vee r_2^{(0)}\right)$. Then we have
$$\pi_1\Psi_\rho(\mu_1,r)e^{-\lambda_*(r)\Psi_\rho(\mu_1,r)}+\pi_2\Psi_\rho(\mu_2,r)e^{-\lambda_*(r)\Psi_\rho(\mu_2,r)}=0$$
by \eqref{eq:constraint-lambda} (with $s=2$), and therefore
$$e^{-\lambda_*(r)\Psi_\rho(\mu_1,r)}\left(\pi_1\Psi_\rho(\mu_1,r)+\pi_2\Psi_\rho(\mu_2,r)e^{-\lambda_*(r)(\Psi_\rho(\mu_2,r)-\Psi_\rho(\mu_1,r))}\right)=0;$$
this yields
$$\pi_1\Psi_\rho(\mu_1,r)+\pi_2\Psi_\rho(\mu_2,r)e^{-\lambda_*(r)(\Psi_\rho(\mu_2,r)-\Psi_\rho(\mu_1,r))}=0,$$
and we get
$$e^{-\lambda_*(r)(\Psi_\rho(\mu_2,r)-\Psi_\rho(\mu_1,r))}=-\frac{\pi_1\Psi_\rho(\mu_1,r)}{\pi_2\Psi_\rho(\mu_2,r)};$$
thus
\begin{equation}\label{eq:lambda*-s=2}
\lambda_*(r)=-\frac{1}{\Psi_\rho(\mu_2,r)-\Psi_\rho(\mu_1,r)}\log\left(-\frac{\pi_1\Psi_\rho(\mu_1,r)}{\pi_2\Psi_\rho(\mu_2,r)}\right).
\end{equation}
We remark that, for
$r\in(\underline{r}_{\rho},\overline{r}_{\rho})$, we have
$$\Psi_\rho(\mu_2,r),\Psi_\rho(\mu_1,r)\neq 0\ \mbox{and}\ \Psi_\rho(\mu_2,r)\Psi_\rho(\mu_1,r)<0;$$
thus, in particular, $\Psi_\rho(\mu_2,r)-\Psi_\rho(\mu_1,r)\neq
0$. Finally, by \eqref{eq:rf-expression-special-constraint} (with
$s=2$) and \eqref{eq:lambda*-s=2}, we get
$$H_{\rho,\langle\pi,\mu\rangle}(r)=
-\log\left(\sum_{h=1}^2\pi_h\left(-\frac{\pi_1\Psi_\rho(\mu_1,r)}{\pi_2\Psi_\rho(\mu_2,r)}\right)
^{\frac{\Psi_\rho(\mu_h,r)}{\Psi_\rho(\mu_2,r)-\Psi_\rho(\mu_1,r)}}\right).$$
\end{remark}

\begin{remark}[On Theorem \ref{prop:main} with $s=\infty$]\label{rem:case-s=infty}
It is possible to present a version of Theorem \ref{prop:main}
with $s=\infty$, namely for the case of countable mixture models.
The statement and the proof can be easily adapted and we omit the
details. However we remark that Condition
\ref{cond:equivalent-formulation-on-constraint} and Remark
\ref{rem:consequence-condition} have to be suitably changed; in
fact we should require that the quantities
$\underline{r}_{\rho}:=\min\left\{r_i^{(0)}:i\geq 1\right\}$ and
$\overline{r}_{\rho}:=\max\left\{r_i^{(0)}:i\geq 1\right\}$ are
well-defined (this is not guaranteed as happens for the case
$s<\infty$).
\end{remark}

Finally we are also interested in the local comparison between
rate functions around the points where they uniquely vanish; in
fact the larger $H_{\rho,\langle\pi,\mu\rangle}$ is around $r_0$
(except $r_0$), the faster is the convergence of estimators to
$r_0$ (as $n\to\infty$). We also recall that, under suitable
hypotheses which guarantee the existence of the second derivative
of $H_{\rho,\langle\pi,\mu\rangle}(r)$ computed at $r=r_0$, the
more $H_{\rho,\langle\pi,\mu\rangle}^{\prime\prime}(r_0)$ is
large, the more $H_{\rho,\langle\pi,\mu\rangle}$ is large around
$r_0$ (except $r_0$). An expression for
$H_{\rho,\langle\pi,\mu\rangle}^{\prime\prime}(r_0)$ is given in
the next proposition.

\begin{proposition}\label{prop:second-derivative}
Consider the same hypotheses and notation of Theorem
\ref{prop:main}(ii). Moreover assume that
$\Psi_\rho(\mu_1,\cdot),\ldots,\Psi_\rho(\mu_s,\cdot)$ are
continuously differentiable functions (at least in a neighborhood
of $r_0$). Then, if we consider the notation
$\Psi_\rho^\prime(\mu_h,r):=\frac{d}{dr}\Psi_\rho(\mu_h,r)$, we
have
$$H_{\rho,\langle\pi,\mu\rangle}^{\prime\prime}(r_0)=\frac{\left(\sum_{h=1}^s\pi_h\Psi_\rho^\prime(\mu_h,r_0)\right)^2}
{\mathrm{Var}_\pi[\Psi_\rho(\mu_\cdot,r_0)]},$$ where (the last
equality holds by \eqref{eq:cond-on-pi-and-r0})
$$\mathrm{Var}_\pi[\Psi_\rho(\mu_\cdot,r_0)]:=\sum_{h=1}^s\pi_h\Psi_\rho^2(\mu_h,r_0)-\left(\sum_{h=1}^s\pi_h\Psi_\rho(\mu_h,r_0)\right)^2
=\sum_{h=1}^s\pi_h\Psi_\rho^2(\mu_h,r_0).$$
\end{proposition}
\begin{proof}
For $r\in(\underline{r}_{\rho},\overline{r}_{\rho})$ we have
\begin{multline*}
H_{\rho,\langle\pi,\mu\rangle}^\prime(r)=\frac{\sum_{h=1}^s\pi_he^{-\lambda_*(r)\Psi_\rho(\mu_h,r)}
[\lambda_*^\prime(r)\Psi_\rho(\mu_h,r)+\lambda_*(r)\Psi_\rho^\prime(\mu_h,r)]}{\sum_{h=1}^s\pi_he^{-\lambda_*(r)\Psi_\rho(\mu_h,r)}}\\
=\lambda_*^\prime(r)\underbrace{\frac{\sum_{h=1}^s\pi_h\Psi_\rho(\mu_h,r)e^{-\lambda_*(r)\Psi_\rho(\mu_h,r)}}
{\sum_{h=1}^s\pi_he^{-\lambda_*(r)\Psi_\rho(\mu_h,r)}}}_{=0\
\mathrm{by}\ \eqref{eq:constraint-lambda}}
+\lambda_*(r)\frac{\sum_{h=1}^s\pi_h\Psi_\rho^\prime(\mu_h,r)e^{-\lambda_*(r)\Psi_\rho(\mu_h,r)}}{\sum_{h=1}^s\pi_he^{-\lambda_*(r)\Psi_\rho(\mu_h,r)}}\\
=\lambda_*(r)\frac{\sum_{h=1}^s\pi_h\Psi_\rho^\prime(\mu_h,r)e^{-\lambda_*(r)\Psi_\rho(\mu_h,r)}}{\sum_{h=1}^s\pi_he^{-\lambda_*(r)\Psi_\rho(\mu_h,r)}};
\end{multline*}
thus, since $\lambda_*(r_0)=0$, we get
$$H_{\rho,\langle\pi,\mu\rangle}^\prime(r_0)=\lambda_*(r_0)\frac{\sum_{h=1}^s\pi_h\Psi_\rho^\prime(\mu_h,r_0)e^{-\lambda_*(r_0)\Psi_\rho(\mu_h,r_0)}}
{\sum_{h=1}^s\pi_he^{-\lambda_*(r_0)\Psi_\rho(\mu_h,r_0)}}=0.$$
Moreover, again for
$r\in(\underline{r}_{\rho},\overline{r}_{\rho})$, we have
$$H_{\rho,\langle\pi,\mu\rangle}^{\prime\prime}(r)
=\lambda_*^\prime(r)\frac{\sum_{h=1}^s\pi_h\Psi_\rho^\prime(\mu_h,r)e^{-\lambda_*(r)\Psi_\rho(\mu_h,r)}}{\sum_{h=1}^s\pi_he^{-\lambda_*(r)\Psi_\rho(\mu_h,r)}}
+\lambda_*(r)\frac{d}{dr}\left(\frac{\sum_{h=1}^s\pi_h\Psi_\rho^\prime(\mu_h,r)e^{-\lambda_*(r)\Psi_\rho(\mu_h,r)}}
{\sum_{h=1}^s\pi_he^{-\lambda_*(r)\Psi_\rho(\mu_h,r)}}\right);$$
thus, by taking into account again $\lambda_*(r_0)=0$, we obtain
\begin{equation}\label{eq:second-derivative-before-IFT}
H_{\rho,\langle\pi,\mu\rangle}^{\prime\prime}(r_0)=\lambda_*^\prime(r_0)\sum_{h=1}^s\pi_h\Psi_\rho^\prime(\mu_h,r_0).
\end{equation}
We conclude by computing $\lambda_*^\prime(r_0)$ by means of the
implicit function theorem. By \eqref{eq:constraint-lambda} we
consider the function
$$\Delta(r,\lambda):=\frac{\sum_{j=1}^s\pi_j\Psi_\rho(\mu_j,r)e^{-\lambda\Psi_\rho(\mu_j,r)}}{\sum_{j=1}^s\pi_je^{-\lambda\Psi_\rho(\mu_j,r)}};$$
the partial derivatives of $\Delta$ are
\begin{multline*}
\Delta_r(r,\lambda)=\frac{1}{(\sum_{j=1}^s\pi_je^{-\lambda\Psi_\rho(\mu_j,r)})^2}
\cdot\left\{\left(\sum_{j=1}^s\pi_j\Psi_\rho^\prime(\mu_j,r)e^{-\lambda\Psi_\rho(\mu_j,r)}(1-\lambda\Psi_\rho(\mu_j,r))\right)
\left(\sum_{j=1}^s\pi_je^{-\lambda\Psi_\rho(\mu_j,r)}\right)\right.\\
\left.+\lambda\left(\sum_{j=1}^s\pi_j\Psi_\rho^\prime(\mu_j,r)e^{-\lambda\Psi_\rho(\mu_j,r)}\right)
\left(\sum_{j=1}^s\pi_j\Psi_\rho(\mu_j,r)e^{-\lambda\Psi_\rho(\mu_j,r)}\right)\right\}
\end{multline*}
and
\begin{multline*}
\Delta_\lambda(r,\lambda)=\frac{1}{(\sum_{j=1}^s\pi_je^{-\lambda\Psi_\rho(\mu_j,r)})^2}
\cdot\left\{-\left(\sum_{j=1}^s\pi_j\Psi_\rho^2(\mu_j,r)e^{-\lambda\Psi_\rho(\mu_j,r)}\right)
\left(\sum_{j=1}^s\pi_je^{-\lambda\Psi_\rho(\mu_j,r)}\right)\right.\\
\left.+\left(\sum_{j=1}^s\pi_j\Psi_\rho(\mu_j,r)e^{-\lambda\Psi_\rho(\mu_j,r)}\right)^2\right\};
\end{multline*}
thus, by taking into account \eqref{eq:cond-on-pi-and-r0} and
$\lambda_*(r_0)=0$, we have
$$\Delta_r(r_0,\lambda_*(r_0))=\sum_{j=1}^s\pi_j\Psi_\rho^\prime(\mu_j,r_0)\
\mbox{and}\
\Delta_\lambda(r_0,\lambda_*(r_0))=-\sum_{j=1}^s\pi_j\Psi_\rho^2(\mu_j,r_0),$$
and the implicit function theorem yields
$$\lambda_*^\prime(r_0)=\left.-\frac{\Delta_r(r,\lambda)}{\Delta_\lambda(r,\lambda)}\right|_{(r,\lambda)=(r_0,\lambda_*(r_0))}
=\frac{\sum_{j=1}^s\pi_j\Psi_\rho^\prime(\mu_j,r_0)}{\sum_{j=1}^s\pi_j\Psi_\rho^2(\mu_j,r_0)}.$$
We conclude the proof by combining this equality and
\eqref{eq:second-derivative-before-IFT}.
\end{proof}

\section{Examples}\label{sec:examples}
In this section we consider some examples of risk measures
satisfying Condition
\ref{cond:equivalent-formulation-on-constraint}. The first example
concerns risk measures that are linearly dependent with respect to
the weights, and we present two specific cases. Other examples
consider quantiles and the class of shortfall risk measures, that
includes the entropic risk measures as a special case. We conclude
the section with two examples: one for an insurance application
with $s=2$, and one where we obtain explicit expressions for
$s=3$. In view of what follows we consider the notation $F_\mu$
for the distribution function associated with the law $\mu$,
namely
$$F_\mu(x):=\mu((-\infty,x]),\quad \mbox{for all}\ x\in\mathbb{R}.$$
We remark that, when we deal with a finite mixture
$\sum_{j=1}^sp_j\mu_j$ of some laws $\mu_1,\ldots,\mu_s$ (for some
$(p_1,\ldots,p_s)\in\Sigma_s$), we have
$F_{\sum_{j=1}^sp_j\mu_j}=\sum_{j=1}^sp_jF_{\mu_j}$.

\begin{example}[Linear dependence with respect to the weights]\label{ex:linear-dependence}
We assume that the function \eqref{eq:function-of-weights}
satisfies the following condition:
$$\rho\left(\sum_{j=1}^sp_j\mu_j\right)=\sum_{j=1}^sp_j\rho(\mu_j)\ \mbox{for all}\ (p_1,\ldots,p_s)\in\Sigma_s.$$
Obviously we have a continuous function. In this case one has
$$\Psi_\rho(\mu_i,r)=\rho(\mu_i)-r$$
which yields $r_i^{(0)}=\rho(\mu_i)$ and
$r_0=\sum_{j=1}^s\pi_j\rho(\mu_j)$; moreover
$$\underline{r}_{\rho}:=\min\{\rho(\mu_i):i\in\{1,\ldots,s\}\}\
\mbox{and}\
\overline{r}_{\rho}:=\max\{\rho(\mu_i):i\in\{1,\ldots,s\}\}.$$ Now
we present some formulas for the rate function
$H_{\rho,\langle\pi,\mu\rangle}$ when
$\underline{r}_{\rho}<\overline{r}_{\rho}$; in view of this we
remark that, for $s=2$, we have
$\underline{r}_{\rho}<\overline{r}_{\rho}$ if and only if
$\rho(\mu_1)\neq\rho(\mu_2)$. By
\eqref{eq:rf-expression-special-constraint}, for
$r\in(\underline{r}_{\rho},\overline{r}_{\rho})$, we get
$$H_{\rho,\langle\pi,\mu\rangle}(r)=-\log\left(\sum_{h=1}^s\pi_he^{-\lambda_*(r)\Psi_\rho(\mu_h,r)}\right)
=-r\lambda_*(r)-\log\left(\sum_{h=1}^s\pi_he^{-\lambda_*(r)\rho(\mu_h)}\right).$$
Moreover, by Remark \ref{rem:case-s=2} concerning the case $s=2$,
for $r\in(\underline{r}_{\rho},\overline{r}_{\rho})$ we get
$$H_{\rho,\langle\pi,\mu\rangle}(r)=
-\log\left(\sum_{h=1}^2\pi_h\left(-\frac{\pi_1(\rho(\mu_1)-r)}{\pi_2(\rho(\mu_2)-r)}\right)
^{\frac{\rho(\mu_h)-r}{\rho(\mu_2)-\rho(\mu_1)}}\right);$$ in
particular we can easily check that this formula yields
$H_{\rho,\langle\pi,\mu\rangle}(r_0)=0$. Finally, by Proposition
\ref{prop:second-derivative} (and after some computations where we
take into account that $r_0=\sum_{j=1}^s\pi_j\rho(\mu_j)$), we get
$$H_{\rho,\langle\pi,\mu\rangle}^{\prime\prime}(r_0)=\frac{1}{\sum_{h=1}^s\pi_h(\rho(\mu_h)-r_0)^2}=\frac{1}{\sum_{h=1}^s\pi_h\rho^2(\mu_h)-r_0^2}.$$
\end{example}

Here we briefly present two particular cases concerning Example
\ref{ex:linear-dependence}.
\begin{itemize}
\item The expected value (when $\mu_1,\ldots,\mu_s$ are probability
measures of integrable random variables); in fact we have
$$\int_\mathbb{R}x\sum_{j=1}^sp_j\mu_j(dx)=\sum_{j=1}^sp_j\int_\mathbb{R}x\mu_j(dx).$$
\item The Expected Shortfall $\mathrm{ES}_\alpha$, for $\alpha\in(0,1)$,
when $\mu_1,\ldots,\mu_s$ have the same $\alpha$-quantile, namely
when
$F_{\mu_1}^{-1}(\alpha)=\cdots=F_{\mu_s}^{-1}(\alpha)=:r_\alpha$.
We recall that
$$\mathrm{ES}_\alpha(\mu):=\frac{1}{1-\alpha}\int_{F_{\mu}^{-1}(\alpha)}^{\infty}x\mu(dx),$$
and that we have
$\left(\sum_{j=1}^sp_jF_{\mu_j}\right)^{-1}(\alpha)=r_\alpha$.
Therefore
\begin{align*}
\mathrm{ES}_\alpha\left(\sum_{j=1}^sp_j\mu_j\right)
&=\frac{1}{1-\alpha}\int_{\left(\sum_{j=1}^sp_jF_{\mu_j}\right)^{-1}(\alpha)}^\infty
x\sum_{j=1}^sp_j\mu_j(dx)\\
&=\sum_{j=1}^s\frac{p_j}{1-\alpha}\int_{r_\alpha}^\infty
x\mu_j(dx)=\sum_{j=1}^sp_j\mathrm{ES}_\alpha(\mu_j).
\end{align*}
\end{itemize}

We conclude with the final examples.

\begin{example}[Quantiles]\label{ex:quantiles}
Let us consider $\alpha\in(0,1)$ and strictly increasing and
continuous distribution functions $F_{\mu_1},\ldots,F_{\mu_s}$ on
the same interval. We assume that the function
\eqref{eq:function-of-weights} is defined by
$$\Sigma_s\ni(p_1,\ldots,p_s)\mapsto\rho\left(\sum_{j=1}^sp_j\mu_j\right):=\left(\sum_{j=1}^sp_jF_{\mu_j}\right)^{-1}(\alpha).$$
This function is continuous (see Appendix for details). In this
case one has
$$\Psi_\rho(\mu_i,r)=\alpha-F_{\mu_i}(r)$$
which yields $r_i^{(0)}=F_{\mu_i}^{-1}(\alpha)$ and
$r_0=\left(\sum_{j=1}^s\pi_jF_{\mu_j}\right)^{-1}(\alpha)$;
moreover
$$\underline{r}_{\rho}:=\min\left\{F_{\mu_i}^{-1}(\alpha):i\in\{1,\ldots,s\}\right\}\
\mbox{and}\
\overline{r}_{\rho}:=\max\left\{F_{\mu_i}^{-1}(\alpha):i\in\{1,\ldots,s\}\right\}.$$
Now we present some formulas for the rate function
$H_{\rho,\langle\pi,\mu\rangle}$ when
$\underline{r}_{\rho}<\overline{r}_{\rho}$; in view of this we
remark that, for $s=2$, we have
$\underline{r}_{\rho}<\overline{r}_{\rho}$ if and only if
$F_{\mu_1}^{-1}(\alpha)\neq F_{\mu_2}^{-1}(\alpha)$. By
\eqref{eq:rf-expression-special-constraint}, for
$r\in(\underline{r}_{\rho},\overline{r}_{\rho})$, we get
$$H_{\rho,\langle\pi,\mu\rangle}(r)=-\log\left(\sum_{h=1}^s\pi_he^{-\lambda_*(r)\Psi_\rho(\mu_h,r)}\right)
=\alpha\lambda_*(r)-\log\left(\sum_{h=1}^s\pi_he^{\lambda_*(r)F_{\mu_h}(r)}\right).$$
Moreover, by Remark \ref{rem:case-s=2} concerning the case $s=2$,
for $r\in(\underline{r}_{\rho},\overline{r}_{\rho})$ we get
$$H_{\rho,\langle\pi,\mu\rangle}(r)=
-\log\left(\sum_{h=1}^2\pi_h\left(-\frac{\pi_1(\alpha-F_{\mu_1}(r))}{\pi_2(\alpha-F_{\mu_2}(r))}\right)
^{\frac{\alpha-F_{\mu_h}(r)}{F_{\mu_1}(r)-F_{\mu_2}(r)}}\right);$$
in particular we can easily check that this formula yields
$H_{\rho,\langle\pi,\mu\rangle}(r_0)=0$ because
$\sum_{j=1}^s\pi_jF_{\mu_j}(r_0)=\alpha$. Finally, by Proposition
\ref{prop:second-derivative} (and after some computations where we
take into account again that
$\sum_{j=1}^s\pi_jF_{\mu_j}(r_0)=\alpha$), we get
$$H_{\rho,\langle\pi,\mu\rangle}^{\prime\prime}(r_0)=\frac{\left(\sum_{h=1}^s\pi_hF_{\mu_i}^\prime(r)\right)^2}{\sum_{h=1}^s\pi_h(\alpha-F_{\mu_h}(r_0))^2}
=\frac{\left(\sum_{h=1}^s\pi_hF_{\mu_i}^\prime(r)\right)^2}{\sum_{h=1}^s\pi_hF_{\mu_h}^2(r_0)-\alpha^2}.$$
\end{example}

\begin{example}[Shortfall risk measures]\label{ex:entropic}
We recall some preliminaries (see \cite{FS02}). Given a
\textit{loss function} $\ell:\mathbb{R}\to\mathbb{R}$ (that is a
convex, increasing and not identically constant function) and an
interior point $x_0$ in the range of $\ell$, a shortfall risk
measure is defined by
$$\rho(\mu):=\inf\left\{m\in\mathbb{R}:\int_{\mathbb{R}}\ell(x-m)\mu(dx)\leq x_0\right\};$$
moreover it is the unique solution $m$ to the equation
$$\int_{\mathbb{R}}\ell(x-m)\mu(dx)=x_0.$$
We assume that the function \eqref{eq:function-of-weights} is
defined by
$$\Sigma_s\ni(p_1,\ldots,p_s)\mapsto\rho\left(\sum_{j=1}^sp_j\mu_j\right):=r,\ \mbox{where}\ \sum_{j=1}^sp_j\int_{\mathbb{R}}\ell(x-r)\mu_j(dx)=x_0.$$
The continuity of this function can be checked by adapting the
proof in the Appendix. In this case one has
$$\Psi_\rho(\mu_i,r)=\int_{\mathbb{R}}\ell(x-r)\mu_i(dx)-x_0.$$

From now on, in order to have explicit results, we continue our
analysis for the class of entropic risk measures, that is the case
where we have the loss function $\ell(x)=e^{\theta x}$, for
$\theta>0$, and $x_0=1$. We can check the following equalities:
$$\rho(\mu)=\frac{1}{\theta}\log\left(\int_{\mathbb{R}}e^{\theta x}\mu(dx)\right);$$
the function \eqref{eq:function-of-weights} becomes
$$\Sigma_s\ni(p_1,\ldots,p_s)\mapsto\rho\left(\sum_{j=1}^sp_j\mu_j\right):=\frac{1}{\theta}\log\int_\mathbb{R}e^{\theta x}\sum_{j=1}^sp_j\mu_j(dx);$$
(so we have
$\Sigma_s\ni(p_1,\ldots,p_s)\mapsto\frac{1}{\theta}\log\left(\sum_{j=1}^sp_je^{\theta\rho(\mu_j)}\right)$,
which is a continuous function); the function $\Psi_\rho(\mu_i,r)$
can be rewritten as
$$\Psi_\rho(\mu_i,r)=e^{\theta\rho(\mu_i)}-e^{\theta r},$$
which yields $r_i^{(0)}=\rho(\mu_i)$ and
$r_0=\frac{1}{\theta}\log\left(\sum_{j=1}^s\pi_je^{\theta\rho(\mu_j)}\right)$;
moreover
$$\underline{r}_{\rho}:=\min\{\rho(\mu_i):i\in\{1,\ldots,s\}\}\
\mbox{and}\
\overline{r}_{\rho}:=\max\{\rho(\mu_i):i\in\{1,\ldots,s\}\}.$$ Now
we present some formulas for the rate function
$H_{\rho,\langle\pi,\mu\rangle}$ when
$\underline{r}_{\rho}<\overline{r}_{\rho}$; in view of this we
remark that, for $s=2$, we have
$\underline{r}_{\rho}<\overline{r}_{\rho}$ if and only if
$\rho(\mu_1)\neq\rho(\mu_2)$. By
\eqref{eq:rf-expression-special-constraint}, for
$r\in(\underline{r}_{\rho},\overline{r}_{\rho})$, we get
$$H_{\rho,\langle\pi,\mu\rangle}(r)=-\log\left(\sum_{h=1}^s\pi_he^{-\lambda_*(r)\Psi_\rho(\mu_h,r)}\right)
=-e^{\theta
r}\lambda_*(r)-\log\left(\sum_{h=1}^s\pi_he^{-\lambda_*(r)e^{\theta\rho(\mu_i)}}\right).$$
Moreover, by Remark \ref{rem:case-s=2} concerning the case $s=2$,
for $r\in(\underline{r}_{\rho},\overline{r}_{\rho})$ we get
$$H_{\rho,\langle\pi,\mu\rangle}(r)=
-\log\left(\sum_{h=1}^2\pi_h\left(-\frac{\pi_1(e^{\theta\rho(\mu_1)}-e^{\theta
r})}{\pi_2(e^{\theta\rho(\mu_2)}-e^{\theta r})}\right)
^{\frac{e^{\theta\rho(\mu_h)}-e^{\theta
r}}{e^{\theta\rho(\mu_2)}-e^{\theta\rho(\mu_1)}}}\right);$$ in
particular we can easily check that this formula yields
$H_{\rho,\langle\pi,\mu\rangle}(r_0)=0$ because $e^{\theta
r_0}=\sum_{j=1}^s\pi_je^{\theta\rho(\mu_j)}$. Finally, by
Proposition \ref{prop:second-derivative} (and after some
computations where we take into account again that $e^{\theta
r_0}=\sum_{j=1}^s\pi_je^{\theta\rho(\mu_j)}$), we get
$$H_{\rho,\langle\pi,\mu\rangle}^{\prime\prime}(r_0)=\frac{(\theta e^{\theta r_0})^2}{\sum_{h=1}^s\pi_h(e^{\theta\rho(\mu_h)}-e^{\theta r_0})^2}
=\frac{(\theta e^{\theta
r_0})^2}{\sum_{h=1}^s\pi_he^{2\theta\rho(\mu_h)}-e^{2\theta
r_0}}.$$
\end{example}
\begin{example}[An insurance example with
$s=2$]\label{ex:insurance-with-s=2} In actuarial science mixture
distributions are particularly relevant for modeling different
claim sizes. Consider, for instance, a car insurance context where
individuals are grouped into $s$ categories depending on their
accident history; we assume for convenience that $s=2$. Each group
claim distribution $\mu_j$ may be modeled using an exponential
distribution with parameter $\lambda_j$, $j\in\{1,2\}$; thus
$F_{\mu_j}(x)=1-e^{-\lambda_j(x)}$, with $x>0$ and the probability
of arrival of a claim in group $j$ is $\pi_j>0$, with
$\pi_1+\pi_2=1$. We assume that the $\lambda_j$'s are given and
without loss of generality $\lambda_1<\lambda_2$ (for
$\lambda_1=\lambda_2$ we find the usual exponential distribution).
Instead the $\pi_j$'s are estimated empirically from  a sample of
$n$ claims; denoting $X_i$ a random variable taking values $1$ or
$2$ depending on whether  claim $i$ belongs to the group $j$, we
obtain $X_i=j$ with probability $\pi_j$ and
$\hat{\pi}_n(j)=\frac{1}{n}\sum_{i=1}^n\delta_{X_i=j}$ for all
$j\in\{1,2\}$. In this example, we are interested in understanding
what happens to $\rho\left(\sum_{j=1}^s\hat{\pi}_n(j)\mu_j\right)$
as $n\to\infty$ and the risk measure $\rho$ is quantile at level
$\alpha\in(0,1)$, also known as Value-at-Risk
($\mathrm{VaR}_\alpha$) in the risk management literature. We
recall that for a model $\mu$ with continuous and strictly
increasing distribution function $F$, we have
$\rho(\mu)=F^{-1}(\alpha)$, therefore we have
$$
\rho(\mu_j)=-\frac{1}{\lambda_j}\log(1-\alpha), \quad j\in\{1,2\}.
$$
From Example \ref{ex:quantiles}, we know that
$$
\Psi_{\rho}(\mu_j,r)=\alpha-F_{\mu_j}(r)=\alpha-1+e^{-\lambda_j r},
$$
and
$$
\overline{r}_\rho=r_1^{(0)}=-\frac{1}{\lambda_1}\log(1-\alpha)>-\frac{1}{\lambda_2}\log(1-\alpha)=r_2^{(0)}=\underline{r}_\rho.
$$
From Remark \ref{rem:case-s=2}, for $r\in(\underline{r}_\rho,\overline{r}_\rho)$ we easily find
$$
\lambda_{*}(r)=-\frac{1}{e^{-\lambda_2r}-e^{-\lambda_1r}}\log\left(-\frac{\pi_1(\alpha-1+e^{-\lambda_1r})}{\pi_2(\alpha-1+e^{-\lambda_2r})}\right)
$$
and the weights $p=(p_1,p_2)$ in \eqref{eq:def-minimiser-special-constraint} which attain the infimum in \eqref{eq:variational-formula-special-constraint} are given by
$$
p_j(r)=\frac{\pi_je^{-\lambda_*(r)(\alpha-1+e^{-\lambda_jr})}}{\sum_{j=1}^2\pi_je^{-\lambda_*(r)(\alpha-1+e^{-\lambda_jr})}},\quad j\in\{1,2\}.
$$
We then obtain the rate function:
$$
H_{\rho,<\pi,\mu>}(r)=-\log\left(\pi_1\left(-\frac{\pi_1(\alpha-1+e^{-\lambda_1r})}{\pi_2(\alpha-1+e^{-\lambda_2r})}\right)^{\frac{\alpha-1+e^{-\lambda_1r}}{e^{-\lambda_2r}-e^{-\lambda_1r}}}+\pi_2\left(-\frac{\pi_1(\alpha-1+e^{-\lambda_1r})}{\pi_2(\alpha-1+e^{-\lambda_2r})}\right)^{\frac{\alpha-1+e^{-\lambda_2r}}{e^{-\lambda_2r}-e^{-\lambda_1r}}}\right)
.$$
We refer the interested reader to \cite{LLW12} for a more detailed analysis of the use of exponential mixture models in insurance.
\end{example}

\begin{example}[A specific example with $s=3$]\label{ex:specific-with-s=3}
We refer to Condition
\ref{cond:equivalent-formulation-on-constraint} with $s=3$ and,
for some $a>0$ (and a suitable strictly decreasing function
$\Psi$), we set
$$\Psi_\rho(\mu_i,r):=\Psi(r)+(i-1)a\ (\mbox{for all}\ i\in\{1,2,3\});$$
(we remark that we are in this situation when we deal with Example
\ref{ex:linear-dependence}, with $s=3$; in such a case we have
$\Psi(r)=\rho(\mu_1)-r$, $\rho(\mu_2)=\rho(\mu_1)+a$, and
$\rho(\mu_3)=\rho(\mu_1)+2a$).

We recall that
$$\underline{r}_{\rho}\leq r_1^{(0)}<r_2^{(0)}<r_3^{(0)}\leq\overline{r}_{\rho},$$
where $r_i^{(0)}=\Psi^{-1}((i-1)a)$ (for all $i\in\{1,2,3\}$).
Moreover, after some computations, \eqref{eq:cond-on-pi-and-r0}
with $s=3$ yields
$$r_0=\Psi^{-1}(-a(\pi_2+2\pi_3)).$$
Now we compute the rate function $H_{\rho,\langle\pi,\mu\rangle}$
in Theorem \ref{prop:main}(ii). We have
$$H_{\rho,\langle\pi,\mu\rangle}(\Psi^{-1}(0))=-\log\pi_1\ \mbox{and}\ H_{\rho,\langle\pi,\mu\rangle}(\Psi^{-1}(-2a))=-\log\pi_3,$$
which concern the cases $r=\underline{r}_{\rho}$ and
$r=\overline{r}_{\rho}$. In what follows we take
$r\in(\underline{r}_{\rho},\overline{r}_{\rho})=(\Psi^{-1}(0),\Psi^{-1}(-2a))$.
Firstly \eqref{eq:constraint-lambda} yields
$$e^{-\lambda_*(r)\Psi(r)}\left(\pi_1\Psi(r)+\pi_2(\Psi(r)+a)e^{-\lambda_*(r)a}+\pi_3(\Psi(r)+2a)e^{-2\lambda_*(r)a}\right)=0,$$
and therefore
$$e^{-\lambda_*(r)a}=\frac{-\pi_2(\Psi(r)+a)+\sqrt{\pi_2^2(\Psi(r)+a)^2-4\pi_1\pi_3\Psi(r)(\Psi(r)+2a)}}{2\pi_3(\Psi(r)+2a)};$$
we remark that
$\pi_2^2(\Psi(r)+a)^2-4\pi_1\pi_3\Psi(r)(\Psi(r)+2a)\geq 0$
because $\Psi(r)<0$ and $\Psi(r)+2a>0$, and
$$\sqrt{\pi_2^2(\Psi(r)+a)^2-4\pi_1\pi_3\Psi(r)(\Psi(r)+2a)}\geq|\pi_2(\Psi(r)+a)|.$$
Thus we easily obtain $\lambda_*(r)$ and the rate function
expression is
$$H_{\rho,\langle\pi,\mu\rangle}(r)=-\log\left(\sum_{j=1}^3
\pi_j\left(\frac{-\pi_2(\Psi(r)+a)+\sqrt{\pi_2^2(\Psi(r)+a)^2-4\pi_1\pi_3\Psi(r)(\Psi(r)+2a)}}{2\pi_3(\Psi(r)+2a)}\right)^{\Psi_\rho(\mu_j,r)/a}\right).$$

A final remark concerns the condition $\lambda_*(r_0)=0$ (see
Remark \ref{rem:not-explicit}). The above formulas yield
$$\frac{-\pi_2(\Psi(r_0)+a)+\sqrt{\pi_2^2(\Psi(r_0)+a)^2-4\pi_1\pi_3\Psi(r_0)(\Psi(r_0)+2a)}}{2\pi_3(\Psi(r_0)+2a)}=1$$
and, after some computations, we get
$$\pi_1\Psi(r_0)+\pi_2(\Psi(r_0)+a)+\pi_3(\Psi(r_0)+2a)=0$$
which recovers \eqref{eq:cond-on-pi-and-r0} (with $s=3$).
\end{example}

\paragraph{Acknowledgements.} The authors  are
grateful to  anonymous reviewers for their careful reports which
highly improved the paper.


\section*{Appendix: The continuity of $\rho\left(\sum_{j=1}^sp_j\mu_j\right)$ in Example \ref{ex:quantiles}}
We remark that, if we set
$\varphi(x,p_1,\ldots,p_s):=\sum_{j=1}^sp_jF_{\mu_j}(x)$, we are
interested in the function
$$x(p_1,\ldots,p_s):=\left(\sum_{j=1}^sp_jF_{\mu_j}\right)^{-1}(\alpha),$$
which is the implicit function defined by the condition
$\varphi(x,p_1,\ldots,p_s)=\alpha$.

We assume that $x(p_1,\ldots,p_s)$ is not continuous at some point
$(q_1,\ldots,q_s)$. Then we can find a sequence
$\{(p_1^{(n)},\ldots,p_s^{(n)}):n\geq 1\}$ which converges to
$(q_1,\ldots,q_s)$, and $\{x(p_1^{(n)},\ldots,p_s^{(n)}):n\geq
1\}$ converges to some limit $\ell$, with $\ell\neq
x(q_1,\ldots,q_s)$; moreover we have
$$\varphi(x(p_1^{(n)},\ldots,p_s^{(n)}),p_1^{(n)},\ldots,p_s^{(n)})=\alpha\ (\mbox{for all}\ n\geq 1)$$
and, if we take the limit as $n\to\infty$, we get
$\sum_{j=1}^sq_jF_{\mu_j}(\ell)=\alpha$ (by the continuity of the
functions $F_{\mu_1},\ldots,F_{\mu_s}$). The last equality yields
$\ell=x(q_1,\ldots,q_s)$ by construction, and this is a
contradiction.

\end{document}